\documentclass[a4paper,onecolumn,11pt,accepted=2019-12-20]{quantumarticle}
\pdfoutput=1
\usepackage[latin1]{inputenc}
\usepackage[english]{babel}
\usepackage[T1]{fontenc}
\usepackage{mathrsfs}
\usepackage{amsfonts}
\usepackage{amssymb}
\usepackage{amsmath}
\usepackage{amsthm}
\usepackage{graphicx}
\usepackage[usenames,dvipsnames]{color}
\usepackage[colorlinks=true,citecolor=blue,linkcolor=magenta]{hyperref}
\usepackage{microtype}
\usepackage{braket}
\usepackage{tabularx,colortbl}
\usepackage{pifont}
\usepackage{mathtools}
\usepackage{color}
\usepackage{dsfont}
\usepackage{algorithm}
\usepackage{chngcntr}
\usepackage{complexity}
\usepackage[table, x11names]{xcolor}
\usepackage{epstopdf}
\usepackage{braket}
\usepackage{caption}
\usepackage{subfig}
\usepackage[titletoc,title]{appendix}
\usepackage[colorinlistoftodos,prependcaption,textsize=tiny]{todonotes}
\usepackage{breqn}
\usepackage[english]{babel}
\usepackage[latin1]{inputenc}
\usepackage{mathrsfs}
\usepackage{amsfonts}
\usepackage{amssymb}
\usepackage{amsmath}
\usepackage{amsthm}
\usepackage{graphicx}
\usepackage[usenames,dvipsnames]{color}
\usepackage[colorlinks=true,citecolor=blue,linkcolor=magenta]{hyperref}
\usepackage{microtype}
\usepackage{braket}
\usepackage{tabularx,colortbl}
\usepackage{pifont}
\usepackage{mathtools}
\usepackage{color}
\usepackage{dsfont}
\usepackage{algorithm}
\usepackage{chngcntr}
\usepackage{complexity}
\usepackage{epstopdf}
\usepackage{braket}
\usepackage{caption}
\usepackage{subfig}
\usepackage[titletoc,title]{appendix}
\usepackage[colorinlistoftodos,prependcaption,textsize=tiny]{todonotes}

\usepackage{tikz}
\usepackage{lipsum}

\usepackage[backend=biber, natbib=true, bibstyle=alphabetic,citestyle=alphabetic, doi=true, sorting=nyt, giveninits]{biblatex}
\newtheorem{theorem}{Theorem}
\newtheorem{corollary}[theorem]{Corollary}
\newtheorem{definition}{Definition}
\newtheorem{lemma}{Lemma}

\newtheorem{remark}{Remark}

\newcommand{\comment}[1]{}
\newcommand{\Norm}[1]{\left\lVert#1\right\rVert}
\newcommand{\norm}[1]{\lVert#1\rVert}
\newcommand{\idty}[1]{\mathbb{1}}
\newcommand{\ovsqrt}[1]{\frac{1}{\sqrt{2}}}

\newcommand{\tr}[1]{\mathrm{Tr}}
\newcommand{\Ord}[1]{O(#1)}

\usepackage{algpseudocode}

\newcommand{\N}{\mathbb{N}}
\newcommand{\CC}{\mathbb{C}}

\newcommand{\Nystrom}{Nystr\"om}
\renewcommand{\tr}[1]{\mathrm{Tr}\left( #1 \right)}

\renewcommand{\dag}{+}

\begin{document}

\title{Approximating Hamiltonian dynamics with the $\quad$ {\Nystrom} method}

\author{Alessandro Rudi}
\email{alessandro.rudi@inria.fr}
\thanks{corresponding author}
\affiliation{INRIA - Sierra project team, Paris, France}

\author{Leonard Wossnig}
\email{l.wossnig@cs.ucl.ac.uk}
\affiliation{Department of Computer Science, University College London, London, United Kingdom and Rahko Ltd., London, United Kingdom}

\author{Carlo Ciliberto}
\email{c.ciliberto@ucl.ac.uk}
\affiliation{Department of Computer Science, University College London, London, United Kingdom}

\author{Andrea Rocchetto}
\email{andrea@cs.utexas.edu}
\affiliation{Department of Computer Science, University of Texas at Austin, Austin, United States, Department of Computer Science, University of Oxford, Oxford, United Kingdom, and Department of Computer Science, University College London, London, United Kingdom}

\author{Massimiliano Pontil}
\email{massimiliano.pontil@iit.it}
\affiliation{Computational Statistics and Machine Learning, IIT, Genoa, Italy}

\author{Simone Severini}
\email{s.severini@ucl.ac.uk}
\affiliation{Department of Computer Science, University College London, London, United Kingdom}

\maketitle

\begin{abstract}
  Simulating the time-evolution of quantum mechanical systems is BQP-hard and expected to be one of the foremost applications of quantum computers. We consider classical algorithms for the approximation of Hamiltonian dynamics using subsampling methods from randomized numerical linear algebra. We derive a simulation technique whose runtime scales polynomially in the number of qubits and the Frobenius norm of the Hamiltonian. As an immediate application, we show that sample based quantum simulation, a type of evolution where the Hamiltonian is a density matrix, can be efficiently classically simulated under specific structural conditions. Our main technical contribution is a randomized algorithm for approximating Hermitian matrix exponentials. The proof leverages a low-rank, symmetric approximation via the \Nystrom{} method. Our results suggest that under strong sampling assumptions there exist classical poly-logarithmic time simulations of quantum computations.

\end{abstract}

\section{Introduction}

Special purpose quantum simulators permit the efficient implementation of unitary dynamics governed by physically meaningful families of Hamiltonians, while the general task is BQP-hard -- since we can implement any quantum computation by a sequence of Hamiltonian evolutions. The properties responsible for efficiency correspond to plausible structural restrictions such as locality of interaction between subsystems~\cite{Lloyd1073}, or sparsity, which permits simulation complexities sub-logarithmic in the inverse error~\cite{aharonov2003adiabatic}. Moreover, sparse operators provide examples of a class of quantum circuits with efficient \textit{weak classical simulation}~\cite{schwarz2013simulating}, whereas general strong simulation is known to be \#P -hard~\cite{V08}. 
%-- being equivalent to exact counting the number of satisfying assignments of 3XOR formulas~\cite{}. \\
The intuition supporting this widely used empirical phenomenon is based on the fact that sparse matrices have good storage requirements and an \emph{easier} combinatorial structure mirroring the small number of physical interactions. 

In this paper, we consider the problem of classical simulation of quantum Hamiltonian dynamics. In particular, we will be interested in the case of time-independent Hamiltonians. In this setting, the problem is mathematically equivalent to the task of approximating the matrix exponential of an Hermitian matrix, \textit{i.e.} the Hamiltonian of the system.

The task is particularly relevant for quantum many-body physics where the exponential scaling of the wave function limits to only a handful of qubits the applicability of ordinary differential equation solvers. In order to simulate bigger system it is thus necessary to leverage specific properties of the system. An important class of methods exploits the entanglement structure of the states to obtain efficient representations, known as tensor networks~\cite{verstraete2008matrix,orus2014practical}, that can then be evolved either via Trotterisation~\cite{vidal2004efficient} or with a time-dependent variational principle~\cite{verstraete2004matrix, haegeman2011time}.

We conclude this introduction by noting that,
from a computational perspective,
the problem of simulating Hamiltonian dynamics is related, but different, to the problem of simulating quantum circuits. In the latter case the task no longer involves the computation of costly matrix exponentials and reduces to approximating the action of a circuit, described as a sequence of unitary operations, on a state vector. Recent techniques developed for this problem involve the notion of stabilizer rank~\cite{bravyi2018simulation} or neural network quantum states~\cite{jonsson2018neural}.

\subsection{Results}

We give a classical, randomized algorithm for the strong simulation of quantum Hamiltonian dynamics.
Recall that in a strong quantum simulation one computes the amplitude of a particular outcome, whereas in a weak simulation one only samples from the output distribution of a quantum circuit. While the first is known to have unconditional and exponential lower bounds~\cite{huang2018explicit}, and it is therefore hard for both classical and quantum computers, the latter can be performed efficiently by a quantum computer (for circuits of polynomial size).
In this sense, our approach is tackling a problem that is harder than the one solved by quantum computers (which can only sample from the probability distribution induced by  
measurements on the output state).
Surprisingly, we are still able to find cases in which the time evolution can be simulated efficiently. 

More specifically, given an efficient classical description of a Hamiltonian $H$ and of an $n=\log_2 N$-qubit quantum state $\psi$, we propose conditions under which the evolution can be simulated efficiently.
The algorithm takes as input an index $i \in \{1,\dots, 2^n \}$ and returns the complex amplitude corresponding to the projection of the state vector in the $i$-th element of the computational basis. 

Our algorithm relies on an input model that is commonly used in quantum simulation (condition 1 and 2) supplemented by a condition that is frequently encountered in the randomised numerical linear algebra literature (condition 3). More specifically:

\begin{enumerate}
  \item We require $H$ to be {\em row-computable}. This condition states that every row of $H$ must have at most a number $s = O(\polylog(N))$ of non-zero entries. Furthermore, there exists a classical efficient algorithm that, given a row-index $i$, outputs a list of the non-zero entries of the row.
  \item We require the initial state $\psi$ to have have at most $q=O(\polylog(N))$ non-zero entries. Furthermore, there exists a classical efficient algorithm that outputs a list of the non-zero entries.
  \item  We require $H$ to be efficiently {\em row-searchable}. This condition informally states that we can efficiently sample randomly selected indices of the rows of $H$ in a way proportional to the norm of the row (general case) or the diagonal element of the row (positive semidefinite case). 

\end{enumerate}

It is important to remark that, although the notion of row-searchability is commonly assumed to hold in the randomised numerical linear algebra literature (see for example Section $4$ in the seminal paper by Frieze, Kannan, and Vempala~\cite{frieze2004fast}), in the context of quantum systems this is a much stronger requirement and cannot be assumed to hold for every matrix. Indeed, whilst for a polynomially sized matrix it is always possible to evaluate all the row-norms efficiently in the number of non-zero entries of the matrix, this is generally no longer the case for the exponentially sized matrices found in quantum mechanics.

Under this input model we prove the following theorem that here we state informally.
\begin{theorem}[Informal statement of the main result]
\label{th:maininformal}
If $H$ is a row-computable, row-searchable Hamiltonian on $n$ qubits, and if $\psi$ is an $n$-qubit quantum state with an efficient classical description then, there exists an algorithm that, with probability $1-\delta$, approximates any chosen amplitude of the state $e^{i H t} \psi$ in time 
$$O\left(sq + \frac{t^9\norm{H}^4_F \norm{H}^7}{\epsilon^4}\left(n + \log\frac{1}{\delta}\right)^2\right),$$
where $\epsilon$ determines the quality of the approximation, $\norm{\cdot}_F$ is the Frobenius norm, $\norm{\cdot}$ is the spectral norm, $s$ is the maximum number of non-zero elements in the rows of $H$, and $q$ is the number of non-zero elements in $\psi$.
\end{theorem}

By analyzing the dependency of the runtime on the Frobenius norm we can determine under which conditions we can obtain efficient Hamiltonian simulations. Informally, we obtain:
\begin{corollary}[Informal]
\label{co:maininformal}
If $H$ is a row-computable, row-searchable Hamiltonian on $n$ qubits such that $\norm{H}^2_F - \frac{1}{N}\tr{H}^2 \leq O(\polylog(N))$, and if $\psi$ is an $n$-qubit quantum state with an efficient classical description, then 
there exists an efficient algorithm that approximates any chosen amplitude of the state $e^{i H t} \psi$.
\end{corollary}

The algorithm proceeds by performing a two steps approximation. First, the Hamiltonian $H$ is approximated in terms of a low rank operator $\widehat{H}$ which is more amenable to computations by sampling the rows with a probability proportional to the row norm. Second, the time evolution $e^{i\widehat{H} t}\psi$ is approximated by a truncated Taylor expansion of the matrix exponential. By combining the structure of $\widehat{H}$ and the spectral properties of the truncated exponential we are able to guarantee that the above procedure can be efficiently performed.

In a more specific way, we separately consider the case of a generic Hamiltonian $H$ and the restricted case of positive semidefinite Hamiltonians, for which a refined analysis is possible. The proposed algorithm leverages on a low-rank approximation of the Hamiltonian $H$ to efficiently approximate the matrix exponential $e^{i Ht}$. Such approximation is performed by randomly sampling $M = O(\polylog(N))$ rows according to the distribution determined by the row-searchability condition and then collating them in a matrix $A\in \mathbb C^{M \times N}$.  

When $H$ is positive semidefinite, we consider the approximation $\widehat{H}= AB^+ A^*$, where $B^+$ is the pseudoinverse of the positive semidefinite matrix $B\in\mathbb C^{M \times M}$ obtained by selecting the rows of $A$ whose indices correspond to those originally sampled for the rows of $H$. The approximation of the time evolution $e^{i\widehat{H}t}\psi$ is then performed in terms of a Taylor expansion of the matrix exponential function, truncated to the $K$-th order. Leveraging the structure of $\widehat{H}$ it is possible to formulate the truncated expansion only in terms of linear operations involving the matrices $A^*A$, $B^+$ and $B$ and the vector $A^*\psi$. Under the row-computability assumption, all these operations can be performed efficiently. 

In the general Hermitian setting, rows of $H$, that are sampled to form $A$, are first rescaled according to their sampling probability. Differently from the positive semidefinite case, we resort to a slightly more involved variant of the approximation of the matrix exponential, in particular we use $\widehat{H}^2 := AA^*$ to approximate $H^2$. This approach involves two auxiliary functions in which the matrix exponential is decomposed and for each of these two functions we evaluate its truncated Taylor expansion. We show that this approximation can be formulated only in terms of linear operations involving $A^*A$ and $A^*\psi$. Again, under the row-computability assumption, these operations can be performed efficiently. 

\subsection{Applications and outlook}

It remains an open question---and an interesting research direction---to determine if there exist physically relevant families of Hamiltonians that respect the requirements for efficient simulation outlined in this paper.
In general terms, our methods appear to function well in cases of ``small variations'' (for example, if all the eigenvalues of the Hamiltonian are contained in an exponentially small band) and in this sense share some similarities with perturbation theory.

As a specific application of our theorem, we propose the case of \textit{sample based Hamiltonian simulation}, that is the simulation of quantum dynamics where the Hamiltonian is a density matrix~\cite{lloyd2014quantum}. This type of simulation has recently found applications in various quantum algorithms for machine learning tasks such as linear regression~\cite{schuld2016prediction} and least squares support vector machines~\cite{rebentrost2014quantum}. Note that when these algorithms are used to analyze classical data, they assume that the data can be efficiently encoded into a density matrix.
Specifically, as a direct consequence of our main theorem, we obtain the following result:
\begin{theorem}[Informal]
\label{th:sampleinformal}
If $\rho$ is a row-computable, row-searchable density matrix and if $\psi$ is an $n$-qubit quantum state with an efficient classical description, then 
there exists an efficient algorithm that approximates any chosen amplitude of the state $e^{i \rho t} \psi$.
\end{theorem}
We remark that all our results do not require the Hamiltonian to be sparse (that is, to have at most $\polylog (N)$ non-zero entries) but only to be row-sparse (that is, to have at most $\polylog (N)$ non-zero entries \emph{per row}). The latter requirement is compatible with matrices that are non-sparse.

Going beyond immediate applications, we believe that the introduction of randomised numerical algebra techniques in quantum mechanics may provide a new direction from which to tackle the quantum many-body problem.
In this spirit, it is relevant to mention that shortly after our preprint was posted on arXiv, Tang~\cite{tang2018quantum} showed that, using a strong sampling data structure---and with a significantly higher polynomial overhead---it is possible to derive a classical poly-logarithmic time algorithm for recommendation problems that nearly matches the asymptotic scaling of the quantum algorithm proposed by Kerenidis and Prakash~\cite{kerenidis2016quantum}. The work by Tang assumes the ability to sample from the probability distribution defined by the squared entries of a matrix divided by the $\ell_2$ norm of the matrix, the so called $\ell_2$-norm sampling. Our row-searchable condition is fundamentally equivalent to this and both our results provide evidence that a careful assessment of the state preparation conditions is fundamental in order to determine whether a quantum algorithm provides a true advantage over a classical variant.

More specifically, \cite{tang2018quantum} defines the sample access to the input in the following way;
\begin{definition}[\cite{tang2018quantum}]
\label{def:Tang}
  We have sample access to data point $x \in \mathbb{C}^N$ if, given an index $i \in [N]$, we can produce independent random samples $i \in [N]$ where $i$ is sampled with probability $|x_i|^2/\norm{x}^2$.
\end{definition}
This definition can be easily extended to the columns (or rows) $H_{:,i}$ ($H_{i,:}$) of a Hamiltonian $H$.
In this case, we sample a column with probability $p(i)=\norm{H_{:,i}}^2/\norm{H}_F^2$ (similarly for a row). 
The sampling scheme we use for Hermitian matrices is equivalent to this one (when the matrix is also PSD we use a more computationally efficient variant). 
Additionally, we also provide a method based on binary-trees to compute the sampling probabilities if these are not given by an oracle but stored in a standard memory model.

Informally, we use the following sample access model;
\begin{definition}[Informal statement of the sampling access model]
\label{def:our}
  We have sample access to the Hermitian matrix $H \in \mathbb{C}^{N \times N}$ if, given an index $i \in [N]$, we can randomly sample independent row indices $i \in [N]$ according to probability $p(i)=\norm{H_{:,i}}^2/\norm{H}_F^2$ in time $O(\mathrm{poly}(\log(N)))$ for a polynomially sparse, or structured matrix $H$.
\end{definition}
Definition~\ref{def:Tang} and Definition~\ref{def:our} are fundamentally the same.
The main difference is that, while we use a traditional memory structure and provide a fast way to calculate the marginals using a binary tree (see Sec.~\ref{sec:sampling}), \cite{tang2018quantum} assumes a memory structure which allows one to sample efficiently according to this distribution.

Finally, it is relevant to mention that the techniques introduced by Tang have been directly applied to the problem of classical Hamiltonian simulation in~\cite{chia2019sampling}.
While the time complexity of this algorithm has a $36$-th power dependency on the Frobenius norm of the Hamiltonian, our algorithm scales with a $4$-th power.

\subsection{Related work}

\subsubsection{Classical and quantum approximation of matrix exponentials}

The problem of matrix exponentiation has been extensively studied in the linear algebra literature~\cite{higham2005scaling,higham2009scaling,al2009new,al2011computing}. Presently, for the case of general Hermitian matrices, there is no known algorithm with a runtime logarithmic in the dimension of the input matrix. Such fast scaling is required for the simulation of quantum dynamics where the dimensions of the matrix that describes the evolution scale exponentially with the number of qubits in the system.

The results described in this paper are based on 
randomized numerical linear algebra techniques. These methods, along with results from spectral graph theory, have lead to a variety of new classical algorithms to approximate matrix exponentials~\cite{drineas2006fast,drineas2011faster,mahoney2011randomized,woodruff2014sketching,rudi2015less}. For specific types of matrices, these techniques give efficient runtimes. For example, Orrecchia \textit{et al.}~\cite{orecchia2012approximating} have demonstrated that the spectral sparsifiers of Spielmann and Teng~\cite{spielman2004nearly,spielman2011spectral} can be used to approximate exponentials of strictly diagonally dominant matrices in time almost linear in the number of non-zero entries of $H$. 

Quantum computers on the other hand can approximate efficiently some kinds of matrix exponentials. In particular, there exist time-efficient quantum algorithms for simulating the dynamic of row-sparse Hamiltonians that have only a linear dependency in the row-sparsity. For an important class of algorithms the simulation exploits an efficient edge-coloring of the graph associated with the Hamiltonian matrix $H$~\cite{childs2003exponential,childs2010simulating}. Once this edge coloring is found, the Hamiltonian can be decomposed into a sum of sparse Hamiltonians assuming that $H$ is sparse. 
It is known that these terms can be simulated separately using the Trotter-Suzuki formula \cite{trotter1959product,suzuki1976generalized}.
Improved methods have been developed which result in algorithms with a runtime of $\tilde{O}(s\ \poly(\log(Ns/\epsilon)))$ where, $s$ is the sparsity, $N$ the dimension and $\epsilon$ the maximum error in the solution \cite{childs2012hamiltonian,berry2015hamiltonian,low2017hamiltonian}.

\subsubsection{Randomized numerical linear algebra and \Nystrom{} methods}

Randomized numerical linear algebra (RandNLA) seeks to solve large-scale linear algebra problems exploiting randomisation as a computational resource. Central is the notion of a sketch.
A \textit{sketch} is a smaller or sparser approximation of the original input instance, such as a matrix of data points, that is used to compute quantities of interest. The review of Drineas and Mahoney covers the main ideas and tools of RandNLA~\cite{drineas2018lectures}. 
As a simple example of the techniques used in RandNLA we present here the case of approximate matrix multiplication via random sampling. In general, given two $n \times n$ square matrices $A$ and $B$, computing $AB$ takes $O(n^3)$ time. Drineas, Kannan, and Mahoney showed that, given an efficient sampling method for the rows and columns of the matrices, it is possible to efficiently approximate $AB$ in time $O(c\, n^2)$, where $c$ is the number of rows and columns randomly sampled from the matrices~\cite{drineas2006fast}. The algorithm is structurally very simple. Draw $c$ random samples of columns of $A$ and rows of $B$ according to a probability distribution $p$. Group the samples into two, properly scaled, smaller matrices $C$ and $R$. When the probability distribution $p$ is chosen appropriately and the columns and rows accordingly re-scaled, it is possible to show that that $\norm{AB - CR}_F = O(\norm{A}_F \norm{B}_F /\sqrt{c})$.

Similar ideas can be used to compute low-rank approximations of matrices that can be then used in a wide range of applications~\cite{drineas2006fast2}.
In general, these methods produce sketches that do not preserve given symmetries of the matrix, a fundamental requirement for applications in quantum mechanics. 
A technique where the symmetry of the sketched matrix is preserved is the \Nystrom{} method, a RandNLA tool developed for the approximation of kernel matrices in statistical learning theory. Roughly speaking, the method allows one to construct a lower dimensional, symmetric, positive semidefinite approximation of a given matrix given a sampling schemes for its columns.
More specifically, let $K \in \mathbb{R}^{n \times n}$ be a symmetric, rank $r$, PSD matrix, $K_{:,j}$ the $j$-th column vector of $K$, and $K_{i,:}$ the $i$-th row vector of $K$.
The singular value decomposition of $K$ is $K = U \Sigma U^{\top}$, where the columns of $U$ are orthogonal and $\Sigma = \operatorname{diag} (\sigma_1, \dots, \sigma_r)$ is the matrix of the singular values $\sigma_1, \dots, \sigma_r$ of $K$.
The Moore-Penrose pseudoinverse of $K$ is $
  K^+ = \sum_{t=1} ^r \sigma_t ^{-1} U_{:,t} {U_{:,t}}^{\top}.$
The \Nystrom{} method finds a low-rank approximation of $K$ that preserves the symmetry and PSD property of the matrix.
Let $C$ denote the $n \times l$ matrix formed by (uniformly) sampling $l \ll n$ columns of $K$, $W$ denote the $l \times l$ matrix consisting of the intersection of these $l$ columns with the
corresponding $l$ rows of $K$, and $W_k$ denote the best rank-$k$ approximation of $W$ 
 \[
 W_k = \operatorname{argmin}_{V\in \mathbb{R}^{l \times l}, \operatorname{rank}(V) = k} \norm{V-W}_F.
 \]
The {\Nystrom} method generates a rank-$k$ approximation $\tilde{K}$ of $K$ for $k < n$ defined by:
$$
\tilde{K}_k = C W_k ^+ C^{\top} \approx K
$$
The running time of the algorithm is $O(nkl)$~\cite{kumar2012sampling}.

The performance of the algorithm can be improved using non-uniform sampling schemes.
For a sampling scheme equivalent to the one we use in this paper we have
\begin{theorem}[Theorem $3$ in~\cite{drineas2005nystrom}]
Let $K$ be a $n \times n$ symmetric PSD matrix, let $k < l$ be a rank parameter, and let $\tilde{K}$ be the approximation constructed using the {\Nystrom} method by sampling $l$ columns with probabilities $p_i = |K_{:,i} |^2 / \norm{K}_F ^2$. Let $\epsilon > 0 $ and $\eta = 1 + \sqrt{8\operatorname{log} (1/ \delta)}$, if $l > 64k \eta^2/\epsilon^4$ then, with probability at least $1-\delta$
  $$
  \Norm{K- \tilde{K}_k}_F \leq \Norm{K- K_k}_F + \epsilon \sum_{i=1} ^n K_{ii} ^2.
  $$  
\end{theorem}

The \Nystrom{} method has proved to be a powerful tool in a range of applications where the matrices are approximately low rank. The method in its present form was developed by Williams and Seeger~\cite{williams2001using} as a sampling-based algorithm to solve regression and classification problems involving Gaussian processes. This problem requires the approximation of symmetric, positive semidefinite matrices that can be well low-rank approximated~\cite{williams2001using,williams2002observations}. The technique proved to be closely related to a method for solving linear integral equations developed by \Nystrom{}~\cite{nystrom1930praktische} and hence the name \Nystrom{} method.

It is worth mentioning the \textit{\Nystrom{} extension}, a further refinement of the technique that has found numerous applications ranging from large-scale machine learning problems, to applications in statistics and signal processing~\cite{williams2001using,williams2002observations,zhang2010clustered,talwalkar2008large,fowlkes2004spectral,kumar2012sampling,belabbas2007fast,belabbas2008sparse,kumar2009sampling,li2010making,mackey2011divide,zhang2010clustered,zhang2008improved}. Typical extensions that substantially improve the performance, e.g. lead to lower reconstruction error, introduce non-uniform importance sampling distributions or random mixing of the input before sampling the columns. 

\subsection{Organization}

Section~\ref{sec:preliminaries} introduces relevant notation and definitions. Section~\ref{sec:sampling} discusses the row-searchability condition and the algorithm to sample efficiently from the rows of the Hamiltonian. Section~\ref{sec:PSD} proves the theorem for the approximate simulation of exponentials of PSD matrices. Section~\ref{sec:Hermitian} proves the theorem for the approximate simulation of exponentials of Hermitian matrices. Section~\ref{sec:applications} discusses applications to the simulation of the evolution of density matrices.

\section{Preliminaries}
\label{sec:preliminaries}

We denote vectors with lower-case letters. For a vector $x\in \mathbb{C}^n$, let $x_i$ denotes the $i$-th element of $x$. A vector is sparse if most of its entries are $0$. For an integer $k$, let $[k]$ denotes the set $\{1,\dots, k\}$.  

For a matrix $A\in \mathbb{C}^{m \times n}$ let $a^j := A_{:,j}$, $j \in [n]$ denote the $j$-th column vector of $A$, $a_i := A_{i,:}$, $i \in [m]$ the $i$-th row vector of $A$, and $a_{ij}=A(i,j)$ the $(i,j)$-th element. We denote by $A_{i:j}$ the sub-matrix of $A$ that contains the rows from $i$ to $j$. 

The \textit{supremum} is denoted as $\mathrm{sup}$ and the \textit{infimum} as $\mathrm{inf}$. For a measure space $(X,\Sigma,\mu)$, and a measurable function $f$ an essential upper bound of $f$ is defined as $U_f^{\text{ess}} := \{ l \in \mathbb R : \mu(f^{-1}(l, \infty) = 0 \}$ if the measurable set $f^{-1} (l, \infty)$ is a set of measure zero, i.e., if $f(x) \leq l$ for almost all $x \in X$.  Then the essential supremum is defined as $\text{ess sup} f := \mathrm{inf}\ U_f^{\text{ess}}$.
We let  the $\mathrm{span}$ of a set $S = \{v_i\}_1^k \subseteq \mathbb C^n$ be defined by
$\mathrm{span}\{S\} := \left\lbrace x \in \mathbb C^n \, | \, \exists \, \{\alpha_i \}_1^k \subseteq \mathbb C\text{ with } x =\sum_{i=1}^k \alpha_i v_i \right\rbrace$. The set is linearly independent if $\sum_i \alpha_i v_i = 0$ if and only if $\alpha_i=0$ for all $i$. 
The range of $A \in \mathbb{C}^{m \times n}$ is defined by $\text{range}(A) = \{ y \in \mathbb R^m : y = A x \text{ for some } x \in \mathbb C^n \} = \mathrm{span}(A^1 ,\ldots, A^n)$. Equivalently the range of $A$ is the set of all linear combinations of the columns of $A$. 
The nullspace $\mathrm{null}(A)$ (or kernel $\mathrm{ker}(A)$) is the set of vectors such that $A v = 0$. Given a set $S = \{ v_i \}_1^k \subseteq \mathbb C^n$. 
The null space of $A$ is $\mathrm{null}(A) = \{ x \in \mathbb R ^b : A x = 0 \}$.

The \textit{rank} of a matrix $A \in \mathbb C^{m \times n}$, $\mathrm{rank}(A)$ is the dimension of $\mathrm{range}(A)$ and is equal to the number of linearly independent columns of $A$; Since this is equal to $\mathrm{rank}(A^T)$ it also equals the number of linearly independent rows of $A$, and satisfies $\mathrm{rank}(A) \leq \min \{m,n\}$. 
The trace of a matrix is the sum of its diagonal elements $\tr{A} = \sum_i a_{ii}$. The support of a vector $\mathrm{supp}(v)$ is the set of indices $i$ such that $v_i =0$ and we call it sparsity of the vector. For a matrix we denote the sparsity as the number of non zero entries, while row or column sparsity refers to the number of non-zero entries per row or column. 
A symmetric matrix $A$ is positive semidefinite (PSD) if all its eigenvalues are non-negative. For a PSD matrix $A$ we write $A \succeq 0$. Similarly $A \succeq B$ is the partial ordering which is equivalent to $A-B \succeq 0$. 

We use the following standard norms. The Frobenius norm $
\norm{A}_F = \sqrt{\sum_{i=1}^m \sum_{j=1}^n A_{ij}^2}$, and the spectral norm $\norm{A} = \sup_{x \in \mathbb C^n, \; x \neq 0} \frac{|A x|}{|x|}$. Note that that $\norm{A}_F^2 = \tr{A^T A} = \tr{A A^T}$. Both norms are submultiplicative and unitarily invariant and they are related to each other as $\norm{A} \leq \norm{A}_F \leq \sqrt{n} \norm{A}$. 

The singular value decomposition of $A$ is $A= U \Sigma V^*$ where $U,V$ are unitary matrices and $U^*$ defines the complex conjugate transpose, also called Hermitian conjugate, of $U$. We denote the pseudo-inverse of a matrix $A$ with singular value decomposition $U \Sigma V^*$ as $A^{\dag}:=V \Sigma^{\dag} U^*$.

\section{From row-searchability to efficient row-sampling}
\label{sec:sampling}

All our algorithms require row-searchable Hamiltonians. In this section we define the row-searchability condition and describe an efficient algorithm to sample from the rows of row-searchable Hamiltonians.

Let $n \in \N$. We first introduce a binary tree of subsets spanning $\{0,1\}^n$. In the following, with abuse of notation, we identify binary tuples with the associated binary number. Let $L$ be a binary string with $|L| \leq n$, where $|L|$ denotes the length of the string. We denote with $S(L)$ the set
\begin{align}
S(L) = \{L\} \times \{0,1\}^{n - |L|} = \{(L_1,\dots, L_{|L|}, v_1,\dots,v_{n - |L|}) ~|~ v_1,\dots,v_{n - |L|} \in \{0,1\} \}.
\end{align}
We are now ready to state the row-searchability property for a matrix $H$.

\begin{definition}[Row-searchability]
  \label{cond:row-search}
  Let $H$ be a Hermitian matrix of dimension $2^n$, for $n \in \N$. $H$ is {\em row-searchable} if, for any binary string $L$ with $|L| \leq n$, it is possible to compute the following quantity in $O(\poly(n))$
  \begin{equation}
    \label{norm_diag}
    w(S(L)) = \sum_{i \in S(L)} h(i,H_{:,i}),
  \end{equation}
  where $h$ is the function computing the weight associated to the $i$-th column $H_{:,i}$.
  For positive semidefinite $H$ we use $h(i,H_{:,i}) = H_{i,i}$. i.e.\ the diagonal element $i$ while for general Hermitian $H$ we use $h(i, H_{:,i}) = \|H_{:,i}\|^2$.
\end{definition}
Row-searchability intuitively works as follows. If we are given a binary tree, where the leaves contain the individual probabilities and the parents at each level contain
the marginal over their children nodes, then we can, for a randomly sampled number in $[0,1]$ traverse this tree in $\log(N)$ time to find the leave node that is sampled, i.e.\ the indices of the column of $H$. More specifically, row-searchability requires the evaluation of $w(S(L)))$ as defined in Eq.~(\ref{norm_diag}) which computes marginals of the diagonal of $H$, where the co-elements, \textit{i.e.} the elements where we are not summing over, are defined by the tuple $L$. Hence, for empty $L$, $w(S(L)) = \tr{H}$. 
\begin{algorithm}[t]
\caption{MATLAB code for the sampling algorithm \label{alg:sampling}}
\begin{flushleft}
{\bf Input:} \texttt{wS(L)} corresponds to the function $w(S(L))$ defined in Eq.~\ref{norm_diag}.\\
{{\bf Output:} $L$ is the sampled row index}
\end{flushleft}
\begin{center}
\begin{verbatim}
L = [];
q = rand()*wS(L);

for i=1:n
  if q >= wS([L 0])
      L = [L 0];
  else
      L = [L 1];
      q = q - wS([L 0]);
  end
end
\end{verbatim}
\end{center}
\end{algorithm}
Note that the function $h$ considered is related to leverage score sampling, which is a common approach in randomized algorithm for linear algebra~\cite{mahoney2011randomized,woodruff2014sketching}.

In Alg.~\ref{alg:sampling} we provide an algorithm, that, given a row-searchable $H$, is able to sample an index with probability $p(j) = h(j, H_{:,j})/w(\{0,1\}^n)$.
Let $q$ be a random number uniformly sampled in $[0,T]$, where $T = w(\{0,1\}^n)$ is the sum of the weights associated to all the rows.
The algorithm uses logarithmic search, starting with $L$ empty and adding iteratively $1$ or $0$, to find the index $L$ such that $w(\{0, \dots, L-1\}) \leq q \leq w(\{0, \dots, L\})$.
The total time required to compute one index, is $O(n Q(n))$ where $Q(n)$ is the maximum time required to compute a $w(S(L))$ for $L \in \{0,1\}^n$.
Note that if $w(S(L))$ can be computed efficiently for any $L \in \{0,1\}^n$ then $Q(n)$ is polynomial and the cost of the sampling procedure will be polynomial.
\begin{remark}[Row-searchability more general than sparsity]
Note that if $H$ has a polynomial number of non-zero elements, then $w(S(L))$ can be always computed in polynomial time. Indeed given $L$, we go through the list of elements describing $H$ and select only the ones whose row-index starts with $L$ and then compute $w(S(L))$, both step requiring polynomial time. However $w(S(L))$ can be computed efficiently even for Hamiltonians that are not polynomially sparse. For example, take the diagonal Hamiltonian defined by $H_{ii} = 1/i$ for $i \in [2^n]$. This $H$ is not polynomially sparse and in particular it has an exponential number of non-zero elements, but still $w(S(L))$ can be computed in polynomial time, here in particular in $O(1)$.
\end{remark}

\section{Algorithm for PSD row-searchable Hermitian matrices}
\label{sec:PSD}

Given a $2^n \times 2^n$ matrix $H \succeq 0$, our goal is to produce an approximation of the state
\begin{equation}\label{eq:true-state}
\psi(t) = \exp{(i H t)} \psi.
\end{equation}
In particular, we will provide an expression of the form $\widehat{\exp}(i \widehat{H} t) \psi$,
where $\widehat{\exp}$ and $\widehat{H}$ are a suitable approximation respectively of the exponential function and $H$.
We give here an algorithm for $H\succeq 0$ which we then generalize in the following section to arbitrary Hermitian $H$, if the row-searchability condition, i.e. condition~\ref{cond:row-search}, is fulfilled.

Let $h$ be the diagonal of the positive semidefinite $H$ and let $t_1,\dots, t_M$, with $M \in \N$ be indices independently sampled with repetition from $\{1,\dots, 2^n\}$, with probabilities
\begin{equation}
  \label{eq:probabilities}
  p(q) = h_q / \sum_i h_i,
\end{equation}
e.g. via Alg.~\ref{alg:sampling}.
Then, define the matrix $B \in \CC^{M \times M}$ such that $B_{i,j} = H_{t_i, t_j}$, for $1 \leq i,j \leq M$. Finally, denote by $A \in \CC^{2^n \times M}$ the matrix $A_{i,j} = H_{i,t_j}$ for $1 \leq i \leq 2^n$ and $1 \leq j \leq M$.
The approximated matrix is defined as $\widehat{H} = A B^{\dag} A^*$, where $(\cdot)^\dag$ is the pseudoinverse.
Let us also define a function $g(x) = (e^{i t x} - 1)/x$. Then we have $e^{i t x} = 1 + g(x)x$. Note that $g$ is an analytic function, in particular,
$$g(x) = \sum_{k\geq 1} \frac{(i t)^k}{k!} x^{k-1}.$$Then
\begin{dmath}
e^{i\widehat{H}t} = I + g(\widehat{H})\widehat{H} = I + g(A B^{\dag} A^*)A B^{\dag} A^* = I + A g(B^{\dag} A^* A) B^{\dag} A^*,
\end{dmath}
where the last step is due to the fact that, given an analytic function $q(x) = \sum_{k \geq 0} \alpha_k x^k$, we have
\begin{align*}
q(A B^{\dag} A^*)A B^{\dag} A^* &= \sum_{k\geq 1} \alpha_k (A B^{\dag} A^*)^{k} A B^{\dag} A^* \\
& = A \sum_{k\geq 1} \alpha_k (B^{\dag} A^* A)^k B^{\dag} A^* = A q(B^{\dag} A^* A) B^{\dag} A^*.
\end{align*}
By writing $D = B^{\dag} A^* A$, the algorithm is now
$$\widehat{\psi}_M(t) = \psi +  A g(D) B^{\dag} A^* \psi.$$
Now we approximate $g$ with $g_K(x)$, which limits the series defining $g$ to the first $K$ terms, for $K \in \N$. Moreover, by rewriting $g_K(D)B^{+}(A^*\psi)$ in an iterative fashion, we have
$$b_j = \frac{(i t)^{K-j}}{(K-j)!}v + D b_{j-1}, \quad v =  B^{\dag} (A^* \psi),$$
with $b_0 = \frac{(i t)^{K}}{K!}v$ and so $b_{K-1} =   g(D) B^{\dag} A^* \psi $. Finally, the new approximate state is given by
\begin{equation}\label{eq:algo}
\widehat{\psi}_{K, M}(t) = \psi +  A b_{K-1}.
\end{equation}
\begin{algorithm}[t]
\caption{MATLAB code for approximating Hamiltonian dynamics when $H$ is PSD\label{alg:Nystrom}}
\begin{flushleft}
{\bf Input:} $M$, $T = t_1,\dots, t_M$ list of indices computed via Alg.~\ref{alg:sampling}. The function \texttt{compute\_H\_subMatrix}, given two lists of indices, computes the associated submatrix of $H$. \texttt{compute\_psi\_subVector}, given a list of indices, computes the associated subvector of $\psi$.\\
{\bf Output:} vector $b$, s.t. $b=e^{iHt}\psi$.
\end{flushleft}
\begin{center}
\begin{verbatim}
B = compute_H_subMatrix(T, T);

D = zeros(M,M);
v = zeros(M,1);
for i=1:(2^n/M)
    E = compute_H_subMatrix((i-1)*M+1:M*i, T);
    D = D + E'*E;
    v = v + E'*compute_Psi_subVector((i-1)*M+1:M*i);
end
u
D = B\D;

b = zeros(M,1);
for j=1:K
    b = (1i*t)^(K-j)/factorial(K-j) * v + D*b;
end
\end{verbatim}
\end{center}
\end{algorithm}
A MATLAB implementation of this procedure is presented in Alg.~\ref{alg:Nystrom}. \\
Let the row sparsity $s$ of $H$ be of order $\poly(n)$. Then the total cost of applying this operator is given by $O(M^2 \poly(n) + KM^2 + M^3) )$ in time, where the terms $M^3$ and $M^2 \poly(n)$ are resulting from the calculation of $D$.
%The total cost is $O(2^n M^2 + M^3 + M^2K)$ in time.
To compute the total cost in space, note that we do not have to save $H$ or $A$ in memory, but only $B, D$ and the vectors $v, b_j$, for a total cost of $O(M^2)$. Indeed $D$ can be computed in the following way. Assuming, without loss of generality, to have $2^n/M \in \N$, then
$$D = B^{-1} \sum_{i=1}^{2^n/M} A_{M(i-1)+1:Mi}^*A_{M(i-1)+1:Mi},$$
where $A_{a:b}$ is the submatrix of $A$ containing the rows from $a$ to $b$. A similar reasoning holds for the computation of the vector $v$.
In this computation we have assumed that the sample probabilities are give to us and that we can efficiently sample from the matrix $H$ according to these probabilities. In order to make our algorithm practical we hence need to give an algorithm for performing the sampling.
This properties of the SPDS case algorithm are summarized in the following theorem.

\begin{theorem}[Algorithm for simulating PSD row-searchable Hermitian matrices]
\label{thm:main_psd}
Let $\epsilon, \delta \in (0,1]$, let $K, M \in \N$ and $t > 0$. Let $H$ be positive semidefinite, where $K$ is the number of terms in the truncated series expansions
of $g(\hat H)$ and $M$ the number of samples we take for the approximation.
Let $\psi(t)$ be the true evolution (Eq.~\ref{eq:true-state}) and let $\widehat{\psi}_{K, M}(t)$ be the output of our Alg.~\ref{alg:Nystrom} (Eq.~\ref{eq:algo}).
When
\begin{equation}
K \geq e\, t\norm{H} + \log\frac{2}{\epsilon},\qquad
\quad M \geq \max\left(405 \tr{H},~ \frac{72\tr{H} t}{\epsilon} \log\frac{36\tr{H} t}{\epsilon\delta}\right),
\end{equation}
then the following holds with probability $1-\delta$,
$$\norm{\psi(t) - \widehat{\psi}_{K,M}(t)} \leq \epsilon.$$
\end{theorem}
Note that with the result above, we have that $\widehat{\psi}_{K,M}(t)$ in Eq.~\eqref{eq:algo} (Alg.~\ref{alg:Nystrom}) approximates $\psi(t)$, with error at most $\epsilon$ and with probability at least $1-\delta$, requiring a computational cost that is $\Ord{\frac{s t^2\tr{H}^2}{\epsilon^2} \log^2\frac{1}{\delta}}$ in time and $\Ord{\frac{t^2\tr{H}^2}{\epsilon^2} \log^2\frac{1}{\delta}}$ in memory.

In the following we now prove the first main result of this work.
To prove Theorem~\ref{thm:main_psd} we need the following lemmas. Lemma~\ref{lemma:base-decomposition} performs a basic decomposition of the error in terms of the distance between $H$ and
the approximation $\widehat{H}$ as well as in terms of the approximation $g_K$ with respect to $g$.
Lemma~\ref{lemma:nystrom-analytic-bound} provides an analytic bound on the distance between $H$ and $\widehat{H}$,
expressed in terms of the expectation of eigenvalues or related matrices which are then concentrated in Lemma~\ref{lemma:nystrom-probabilistic-bound}.

\begin{lemma}\label{lemma:base-decomposition}
Let $K, M \in \N$ and $t > 0$, then
$$\norm{\psi(t) - \widehat{\psi}_{K, M}(t)} \leq t \norm{H - \widehat{H}} + \frac{(t\norm{\widehat{H}})^{K+1}}{(K+1)!}.$$
\end{lemma}
\begin{proof}
By definition we have that $e^{ixt} = 1 + g(x)x$ with
$g(x) = \sum_{k \geq 1} x^{k-1}(i t)^k/k!$ and $g_K$ is the truncated version of $g$. By adding and subtracting $e^{i \widehat{H} t}$, we have
\begin{equation}
\norm{e^{iHt}\psi - (I+g_K(\widehat{H})\widehat{H})\psi}
\leq \norm{\psi}\ (\norm{e^{iHt} - e^{i\widehat{H}t}} + \norm{e^{i\widehat{H}t} - (I + g_K(\widehat{H})\widehat{H})}).
\end{equation}
By \cite{nakamoto2003norm},
\begin{dmath}
\norm{e^{iHt} - e^{i\widehat{H}t}} \leq t \norm{H - \widehat{H}},
\end{dmath}
moreover, by \cite{mathias1993approximation}, and since $\widehat{H}$ is Hermitian and hence all the eigenvalues are real, we have
\begin{dmath}
\norm{e^{i\widehat{H}t} - (I + g_K(\widehat{H})\widehat{H})} \leq \frac{(t\norm{\widehat{H}})^{K+1}}{(K+1)!}\sup_{l \in [0,1]}\norm{i^{K+1}e^{i l \widehat{ H}t}} \leq \frac{(t\norm{\widehat{H}})^{K+1}}{(K+1)!}.
\end{dmath}
Finally note that $\norm{\psi} = 1$.
\end{proof}
To study the norm $\norm{H - \widehat{H}}$ note that, since $H$ is positive semidefinite, there exists an operator $S$ such that $H = S S^*$, so $H_{i,j} = s_i^* s_j$ with $s_i,s_j$ the $i$-th and $j$-th row of $S$.
Denote with $C$ and $\widetilde{C}$ the operators
$$C = S^*S, \quad \widetilde{C} = \frac{1}{M}\sum_{j=1}^M \frac{\tr{H}}{h_{t_j}} s_{t_j}s_{t_j}^*.$$
We then obtain the following result.
\begin{lemma}\label{lemma:nystrom-analytic-bound}
The following holds with probability $1$. For any $\tau > 0$,
\begin{equation}
\norm{H - \widehat{H}} \leq \frac{\tau}{1-\beta(\tau)}, \quad
\quad \beta(\tau) = \lambda_{\max}((C+\tau I)^{-1/2}(C- \widetilde{C})(C + \tau I)^{-1/2}),
\end{equation}
moreover $\norm{\widehat{H}} \leq \norm{H}$.
\end{lemma}
\begin{proof}
Define the selection matrix $V \in \CC^{M\times 2^n}$, that is always zero except for one element in each row which is $V_{j, t_j} = 1$ for $1 \leq j \leq M$. Then we have that
$$A = HV^*, \quad B = V H V^*,$$
i.e., $A$ is again given by the rows according to the sampled indiced $t_1,\ldots, t_M$ and $B$ is the submatrix obtained from taking the rows and columns according to the same indices.
In particular by denoting with $\widehat{P}$ the operator $\widehat{P} = S^*V^*(VSS^*V^*)^\dag VS$, and recalling that $H=SS^*$ and $C=S^*S$, we have
$$\widehat{H} = A B^\dag A^* = SS^*V^*(VSS^*V^*)^\dag VSS^* = S \widehat{P} S^*.$$
By definition $\widehat{P}$ is an orthogonal projection operator, indeed it is symmetric and, by definition $Q^\dag Q Q^\dag = Q^\dag$, for any matrix $Q$, then
\begin{dmath}
\widehat{P}^2  = S^*V^*[(VSS^*V^*)^\dag (VSS^*V^*)(VSS^*V^*)^\dag] VS = S^*V^* (VSS^*V^*)^\dag  VS = \widehat{P}.
\end{dmath}
Indeed this is a projection in the row space of the matrix $R:=S^*V^*$, since with the singular value decomposition $R := U_R \Sigma_R V_R^*$ we have
\begin{dmath}
\hat{P} = R^*(R^*R)^\dag R = V_R \Sigma_R U_R^* U_R \Sigma^{-2}_R U_R^* U_R \Sigma_R V_R^* = V_RV_R^*,
\end{dmath}
which spans the same space as $R$.
Finally, since $(I - \widehat{P}) = (I - \widehat{P})^2$, and $\norm{Z^*Z} = \norm{Z}^2$, we have
\begin{dmath}
\norm{H - \widehat{H}} = \norm{S(I - \widehat{P})S^*} = \norm{S(I - \widehat{P})^2S^*} = \norm{(I - \widehat{P})S^*}^2.
\end{dmath}
Note that $\widetilde{C}$ can be rewritten as $\widetilde{C} = S^*V^*LVS$, with $L$ a diagonal matrix, with $L_{jj} = \frac{\tr{H}}{M h_{t_j}}$. Moreover $t_j$ is sampled from the probability $p(q) = h_q/\tr{H}$, so $h_{t_j} > 0$ with probability 1, then $L$ has a finite and strictly positive diagonal, so $\widetilde{C}$ has the same range of $\widehat{P}$.
Now, with  $C = S^*S$, we are able to apply Proposition 3 and Proposition 7 of \cite{rudi2015lessArxiv}, and obtain
\begin{equation}
\norm{(I - \widehat{P})S^*}^2 \leq \frac{\tau}{1-\beta(\tau)}, \quad
\quad \beta(\tau) = \lambda_{\max}((C+\tau I)^{-1/2}(C- \widetilde{C})(C + \tau I)^{-1/2}).
\end{equation}
Finally, note that, since $\widehat{P}$ is a projection operator we have that $\norm{\widehat{P}} = 1$, so
$$\norm{\widehat{H}} = \norm{S \widehat{P} S^*} \leq \norm{\widehat{P}}\norm{S}^2 \leq \norm{S}^2 = \norm{H},$$
where the last step is due to the fact that $H = SS^*$.
\end{proof}

\begin{lemma}\label{lemma:nystrom-probabilistic-bound}
Let $\delta \in (0,1]$ and $\tau > 0$.
When
\begin{equation}
M \geq \max\left(405 \tr{H},~ 67 \tr{H} \log \frac{\tr{H}}{2\delta}\right), \quad
\quad \tau = \frac{9\tr{H}}{M} \log \frac{M}{2\delta},
\end{equation}
then with probability $1-\delta$ it holds that
$$\lambda_{\max}((C+\tau I)^{-1/2}(C- \widetilde{C})(C + \tau I)^{-1/2}) \leq \frac{1}{2}.$$
\end{lemma}
\begin{proof}
Define the random variable $\zeta_j = \sqrt{\frac{\tr{H}}{h_{t_j}}} s_{t_j},$
for $1 \leq j \leq M$. Note that
$$\norm{\zeta_j} \leq \sqrt{\frac{\tr{H}}{h_{t_j}}}\norm{s_{t_j}} \leq \sqrt{\tr{H}},$$
almost surely. Moreover,
\begin{equation}
\mathbb{E} \zeta_j \zeta_j^* = \sum_{q=1}^{2^n} p(q) \frac{\tr{H}}{h_q} s_{q} s_q^* 
= \sum_{q=1}^{2^n} s_{q} s_q^*
= S^*S
= C.
\end{equation}
By definition of $\zeta_j$, we have
$$ \widetilde{C} = \frac{1}{M} \sum_{j=1}^M \zeta_j \zeta_j^*.$$
Since $\zeta_j$ are independent for $1 \leq j \leq M$, uniformly bounded, with expectation equal to $C$, and with $\zeta_j^* (C + \tau I)^{-1} \zeta_j \leq \norm{\zeta_j}^2 \tau^{-1} \leq \tr{H}\tau^{-1}$, we can apply Proposition 8 of \cite{rudi2015lessArxiv}, that uses non-commutative Bernstein inequality for linear operators \cite{tropp2012user}, and obtain
\begin{dmath}
\lambda_{\max}((C+\tau I)^{-1/2}(C- \widetilde{C})(C + \tau I)^{-1/2}) \leq \frac{2 \alpha}{3 M} + \sqrt{\frac{2\alpha}{M t}},
\end{dmath}
with probability at least $1 - \delta$, with $\alpha = \log \frac{4 \tr C}{\tau \delta}$.
Since $$\tr{C} = \tr{S^*S} = \tr{SS^*} = \tr{H},$$ by Remark~1 of \cite{rudi2015lessArxiv}, we have that
$$\lambda_{\max}((C+\tau I)^{-1/2}(C- \widetilde{C})(C + \tau I)^{-1/2}) \leq \frac{1}{2},$$
with probability $1-\delta$, when
$M \geq \max(405 \kappa^2, 67 \kappa^2 \log \frac{\kappa^2}{2\delta})$ and $\tau$ satisfies $\frac{9\kappa^2}{M} \log \frac{M}{2\delta} \leq \tau \leq \norm{C}$ (note that $\norm{C} = \norm{H}$), where $\kappa^2$ is a bound for the following quantity
\begin{equation}
\begin{split}
\inf_{\tau > 0}[(\norm{C} + \tau) (\mathrm{ess}\sup \zeta_j^* (C + \tau I)^{-1} \zeta_j)] \\
\leq \tr{H} \inf_{\tau > 0} \frac{\norm{H} + \tau}{\tau} \leq \tr{H} := \kappa^2,
\end{split}
\end{equation}
where $\mathrm{ess}\sup$ here denotes the essential supremum.
\end{proof}
Now we are ready to prove the Theorem for PSD matrices.
\begin{proof}[Proof of Theorem~\ref{thm:main_psd}]
By Lemma~\ref{lemma:base-decomposition}, we have
$$\norm{e^{iHt}\psi - (I+g_K(\widehat{H})\widehat{H})\psi} \leq t \norm{H - \widehat{H}} + \frac{(t\norm{\widehat{H}})^{K+1}}{(K+1)!}.$$
Let $\tau > 0$. By Lemma~\ref{lemma:nystrom-analytic-bound}, we know that $\norm{\widehat{H}} \leq \norm{H}$ and that
$$\norm{H - \widehat{H}} \leq \frac{\tau}{1-\beta(\tau)},$$
$$\quad \beta(\tau) = \lambda_{\max}((C+\tau I)^{-1/2}(C- \widetilde{C})(C + \tau I)^{-1/2}),$$
with probability $1$.
Finally by Lemma~\ref{lemma:nystrom-probabilistic-bound}, we have that the following holds with probability $1-\delta$,
$$\lambda_{\max}((C+\tau I)^{-1/2}(C- \widetilde{C})(C + \tau I)^{-1/2}) \leq \frac{1}{2},$$
when
$$ M \geq \max\left(405 \tr{H},~ 67 \tr{H} \log \frac{\tr{H}}{2\delta}\right)$$
and
$$\quad \tau = \frac{9\tr{H}}{M} \log \frac{\tr{H}}{2\delta}.$$
So we have
$$ \norm{e^{iHt}\psi - (I+g_K(\widehat{H})\widehat{H})\psi} \leq \frac{18 \tr{H}t}{M} \log\frac{M}{2\delta} + \frac{(t\norm{H})^{K+1}}{(K+1)!},$$
with probability $1-\delta$.\\
Now we select $K$ such that $\frac{(t\norm{H})^{K+1}}{(K+1)!} \leq \frac{\epsilon}{2}$. Since, by the Stirling approximation, we have
$$(K+1)! \geq \sqrt{2\pi} (K+1)^{K+3/2} e^{-K-1} \geq  (K+1)^{K+1} e^{-K-1}.$$
Since $$(1+x)\log(1/(1+x)) \leq -x,$$ for $x > 0$ we can select $K = e t \norm{H} + \log\frac{2}{\epsilon} - 1$, such that we have

\begin{align*}
\log\left(\frac{(t\norm{H})^{K+1}}{(K+1)!}\right) &\leq (K+1) \log \frac{e t \norm{H}}{K+1} \\
&\leq e t \norm{H}\left(1 + \frac{\log\frac{2}{\epsilon}}{e t \norm{H}}\right) \log \frac{1}{1 + \frac{\log\frac{2}{\epsilon}}{e t \norm{H}}} \\
&\leq \log\frac{\epsilon}{2}.
\end{align*}
Finally we require $M$, such that $$\frac{18 \tr{H}t}{M} \log\frac{M}{2\delta} \leq \frac{\epsilon}{2},$$ and select $$M = \frac{72\tr{H} t}{\epsilon} \log\frac{36\tr{H} t}{\epsilon\delta}.$$ 
Then we have that
\begin{equation*}
\frac{18 \tr{H}t}{M} \log\frac{M}{2\delta}
\leq \frac{\epsilon}{2} \frac{\log\frac{36\tr{H} t}{\epsilon\delta} +  \log\log\frac{36\tr{H} t}{\epsilon\delta}}{2 \log\frac{36\tr{H} t}{\epsilon\delta}}
\leq \frac{\epsilon}{2}.
\end{equation*}

\end{proof}

We now extend this result to the more general case of arbitrary Hermitian matrices that fulfill the row-searchability condition, i.e.\ the ability to sample according to some leverage of the rows. This will lead to the second main result of this work.

\section{Algorithm for row-searchable Hermitian matrices}
\label{sec:Hermitian}

In this section we provide the algorithm for simulating Hermitian (possibly non-psd) matrices and we provide guarantees on the efficiency when $H$ is row-searchable.
Let in the following $s$ be the maximum number of non-zero elements of the rows of $H$, $\epsilon$ the error in the approximation of the output states of the algorithm w.r.t.\ the ideal $\psi(t)$, and
$t$ the evolution time of the simulation. Let further $K$ be the order of the truncated series expansions and $M$ the number of samples we take for the approximation. As before we first outline the algorithm and then prove its properties.

For arbitrary matrices $H$ we will use the following algorithm.
Sample $M \in \N$ independent indices $t_1, \dots t_M$, with probability $p(i) = \frac{\|h_i\|^2}{\|H\|_F^2}$, $1 \leq i \leq 2^n$, where $h_i$ is the $i$-th row of $H$ (sample via Alg.~\ref{alg:sampling}).
Then denote with $A$, the matrix $2^n \times M$ defined as
$$A = \left[\frac{1}{\sqrt{Mp(t_1)}}h_{t_1}, \dots, \frac{1}{\sqrt{Mp(t_1)}} h_{t_M}\right].$$
Then we will use $H^2=AA^*$ as the approximation for the Hamiltonian.\\
%for which we will establish the correctness in the following.\\
Define two functions that we will use to approximate $e^{ix}$,
$$f(x) = \frac{\cos(\sqrt{x}) - 1}{x}, \quad g(x) = \frac{\sin(\sqrt{x}) - \sqrt{x}}{x\sqrt{x}},$$
moreover denote with $f_K$ and $g_K$ the $K$-truncated Taylor expansions of $f$ and $g$, for $K \in \N$
$$f_K(x) = \sum_{j=0}^K \frac{(-1)^{j+1}x^j}{(2j+2)!}, \quad g_K(x) = \sum_{j=0}^K \frac{(-1)^{j+1}x^j}{(2j+3)!}.$$
In particular note that
$$e^{ix} = 1 + ix + f(x^2)x^2 + i g(x^2)x^3.$$
Analogously to the previous algorithm, we hence estimate $e^{ix}$ via $f_K$ and $g_K$. The final algorithm will be
\begin{equation}\label{eq:algo-frob}
\widehat{\psi}_{K,M}(t) = \psi + i t u + t^2A f_K(t^2A^*A)v + i t^3 A g_K(t^2A^*A) z,
\end{equation}
with $u = \hat H\psi$, $v = A^*\psi$, $z = A^*u$.
Note that the product $f_k(A^*A)v$ and  $A g_K(A^*A) z$ are done by exploiting the Taylor series form of the two functions and performing only matrix vector products in the same way as in Alg.~\ref{alg:Nystrom}. Denote with $s$ the maximum number of non-zero elements in the rows of $H$, with $q$ the number of non-zero elements in $\psi$.
The final algorithm requires $O(sq)$ in space and time to compute $u$, $O(M\min(s,q))$ in time and $O(M)$ in space to compute $v$ and $O(Ms)$ in time and space to compute $z$.
We therefore obtain a total computational complexity of
\begin{align}
&{\rm time:}~~O\left(sq + M\min(s,q) + sMK\right), \\
&\quad {\rm space:}~~O\left(s(q + M)\right).
\end{align}

Note that if $s > M$ is it possible to further reduce the memory requirements at the cost of more computational time, by computing $B = A^*A$ that can be done in blocks and require $O(sM^2)$ in time and $O(M^2)$ in memory, and then compute
$$\widehat{\psi}_{K,M}(t) = \psi + i t u + t^2A f_K(t^2B)v + i t^3 A g_K(t^2B) z.
$$
In that case the computational cost would be
\begin{align}
\label{eq:Timquation}
&{\rm time:}~~O\left(sq + M\min(s,q) + M^2(s +K)\right), \\
\label{eq:Spquation}
&\quad {\rm space:}~~O\left(sq + M^2\right).
\end{align}
The properties of the this algorithm are summarized in the following theorem (this is a formal statement of Theorem~\ref{th:maininformal}): 

\begin{theorem}[Algorithm for simulating row-samplable Hermitian matrices]
\label{thm:main}
Let $\delta, \epsilon \in (0,1]$. Let $t > 0$ and $K, M \in \mathbb{N}$, where $K$ is the number of terms in the truncated series expansions
of $g(\widehat{H})$ and $M$ the number of samples we take for the approximation, and let $t > 0$.
Let $\psi(t)$ be the true evolution (Eq.~\ref{eq:true-state}) and let $\widehat{\psi}_{K,M}(t)$ be computed as in Eq.~\ref{eq:algo-frob}. When
\begin{align}
\label{eq:Mquation}
&M \geq \frac{256t^4(1+t^2\norm{H}^2)\norm{H}^2_F \norm{H}^2}{\epsilon^2} \log \frac{4\norm{H}^2_F}{\delta\norm{H}^2},\\
\label{eq:Kquation}
&\quad K \geq 4t \sqrt{\norm{H}^2 + \epsilon} + \log \frac{4(1+t\norm{H})}{\epsilon},
\end{align}
then
$$\norm{\widehat{\psi}_{K,M}(t) - \psi(t)} \leq \epsilon,$$
with probability at least $1-\delta$.
\end{theorem}

Note that with the result above, we have that $\widehat{\psi}_{K,M}(t)$ in Eq.~\eqref{eq:algo-frob} approximates $\psi(t)$, with error at most $\epsilon$ and with probability at least $1-\delta$, requiring a computational cost that is $O\left(sq + M\min(s,q) + M^2(s +K)\right)$ in time and $O\left(sq + M^2\right)$ is memory.

Combining Eq.~\ref{eq:Timquation}  and~\ref{eq:Spquation} with Eq.~\ref{eq:Mquation} and~\ref{eq:Kquation}, the whole computational complexity of the algorithm described in this section, is

\begin{align}
&{\rm time:}~~O\left(sq + \frac{t^9\norm{H}^4_F \norm{H}^7}{\epsilon^4}\left(n + \log\frac{1}{\delta}\right)^2\right),\\
&{\rm space:}~~O\left(sq + \frac{t^8\norm{H}^4_F \norm{H}^6}{\epsilon^4}\left(n + \log\frac{1}{\delta}\right)^2\right),
\end{align}
where the quantity $\log \frac{4\norm{H}^2_F}{\delta\norm{H}^2}$ in Eq.~\ref{eq:Mquation} was bounded using the following inequality
$$
\log \frac{\norm{H}^2_F}{\norm{H}^2} \leq \log \frac{2^n \lambda_{MAX} ^2}{\lambda_{MAX} ^2} = n,
$$
where $\lambda_{MAX}$ is the biggest eigenvalue of $H$.

Observe now that simulation of the time evolution of $\alpha I$ does only change the phase of the time evolution, where $I \in \mathbb C^{N \times N}$ is the identity matrix and $\alpha$ some real parameter. We can hence perform the time evolution of $\tilde{H} := H - \alpha I$, since for any efficient classical description of the input state we can apply the time evolution of the diagonal matrix $e^{-i\alpha I t}$. We can then optimize the parameter $\alpha$ such that the Frobenius norm of the operator $\tilde{H}$ is minimized, i.e.
\begin{align}
 \alpha =\underset{\alpha}{\text{argmin}} \norm{\tilde H}^2_F = \underset{\alpha}{\text{argmin}} \norm{H - \alpha I}^2_F,
\end{align}
from which we obtain the condition $\alpha = \frac{\tr{H}}{2^n}$. Since our algorithm requires that $\norm{\tilde H}_F$ is bounded by $\polylog N$. Using the spectral theorem, and the fact that the Frobenius norm is unitarily invariant, this in turn gives us after a bit of algebra the condition
\begin{align}
\label{eq:new_bound}
\norm{H}_F^2 - \frac{1}{N}\tr{H}^2 \leq \Ord{\polylog(N)},
\end{align}
for which we can simulate the Hamiltonian $H$ efficiently.\\
We now prove the second main result of this work and establish the correctness of the above results.

\begin{proof}[Proof of Theorem~\ref{thm:main}]
Denote with
\begin{align*}
\widehat{Z}_K(Ht,At) &= I + i t H + t^2A f_K(t^2A^*A)A^* + i t^3 A g_K(t^2A^*A) A^*H, \\
\widehat{Z}(Ht, At) &= I + i t H + t^2A f(t^2A^*A)A^* + i t^3 A g(t^2A^*A) A^*H.
\end{align*}
By definition of $\widehat{\psi}_{K,M}(t)$ and the fact that $\norm{\psi} = 1$, we have
\begin{dmath*}
\norm{\widehat{\psi}_{K,M}(t) - \psi(t)}  \leq\norm{\widehat{Z}_K(At,Ht) - e^{iHt}}\norm{\psi} \\
\leq \norm{\widehat{Z}_K(At,Ht) - \widehat{Z}(Ht, At)} +  \norm{\widehat{Z}(Ht, At) - e^{iHt}}.
\end{dmath*}
We first study $\norm{\widehat{Z}(Ht, At) - e^{iHt}}$. Define $l(x) = f(x) x $ and $m(x) = g(x) x$. Note that, by the spectral theorem, we have
\begin{dmath*}
\widehat{Z}(Ht, At) = I + i t H + t^2A f(t^2A^*A)A^* + i t^3 A g(t^2A^*A) A^*H
 = I + i t H + t^2 f(t^2AA^*)AA^* + i t^3 g(t^2AA^*) AA^*H
 = I + i t H + l(t^2 AA^*) + i t m(t^2AA^*) H.
\end{dmath*}
Since $$e^{ixt} = 1 + ixt + l(t^2x^2) + it m(t^2x^2)x,$$ we have
\begin{dmath*}
\norm{\widehat{Z}(Ht, At) - e^{iHt}} = \norm{l(t^2 AA^*) - l(t^2H^2) + it m(t^2 AA^*) H - itm(t^2 H^2) H}
 \leq \norm{l(t^2 AA^*) - l(t^2H^2)} + t \norm{m(t^2 AA^*) - m(t^2 H^2)}\norm{H}.
\end{dmath*}
To bound the norms in $l, m$ we will apply Thm.~1.4.1 of \cite{aleksandrov2016operator}.
The theorem state that if a function $f \in L^\infty(\mathbb{R})$, i.e.\ $f$ is in the function space which elements are the essentially bounded measurable functions,
it is entirely on $\mathbb{C}$ and satisfies $|f(z)| \leq e^{\sigma |z|}$ for any $z \in \mathbb{C}$. Then
$\norm{f(A) - f(B)} \leq \sigma \norm{f}_{L^\infty(\mathbb{R})} \norm{A - B}$.
Note that
\begin{align*}
|l(z)| &= \left|\sum_{j=1}^\infty  \frac{(-1)^j z^j}{(2j)!}\right| \leq \sum_{j=1}^\infty  \frac{|z|^j}{(2j)!} \leq \sum_{j=1}^\infty  \frac{|z|^j}{j!} \leq e^{|z|}, \\
|l(z)| &= \left|\sum_{j=1}^\infty  \frac{(-1)^j z^{j}}{(2j+1)!}\right| \leq \sum_{j=1}^\infty  \frac{|z|^{j}}{(2j+1)!} \leq \sum_{j=1}^\infty  \frac{|z|^j}{j!} \leq e^{|z|}.
\end{align*}
Moreover it is easy to see that $\norm{l}_{L^{\infty}(\mathbb{R})}, \norm{m}_{L^{\infty}(\mathbb{R})} \leq 2.$ So
\begin{dmath*}
\norm{\widehat{Z}(Ht, At) - e^{iHt}} \leq 2(1+t\norm{H}) \norm{t^2AA^* - t^2H^2} = 2t^2(1+t\norm{H}) \norm{AA^* - H^2}.
\end{dmath*}
Now note that, by defining the random variable $\zeta_i = \frac{1}{p(t_i)} h_{t_i} h_{t_i}^*,$ we have that
\begin{align*}
&AA^* = \frac{1}{M} \sum_{i=1}^M \zeta_i, \\
& \mathbb{E} [\zeta_i] = \sum_{q=1}^{2^n} p(q) \frac{1}{p(q)} h_{t_i} h_{t_i}^* = H^2, ~~ \forall i.
\end{align*}
Let $\tau > 0$. By applying Thm.~1 of \cite{hsu2014weighted} (or Prop.~9 in \cite{rudi2015lessArxiv}), for which
$$\norm{AA^* - H^2} \leq \sqrt{\frac{\norm{H}^2_F \norm{H}^2 \tau}{M}} + \frac{2\norm{H}^2_F \norm{H}^2 \tau}{M},$$
with probability at least $1 - 4\frac{\norm{H}^2_F}{\norm{H}^2}\tau/(e^\tau - \tau - 1)$. Now since $$1 - 4\frac{\norm{H}^2_F}{\norm{H}^2}\tau/(e^\tau - \tau - 1) \geq 1 - e^{\tau},$$ when $\tau \geq e$, by selecting $$\tau = 2\log\frac{4\norm{H}^2_F}{\norm{H}\delta},$$
we have that the equation above holds with probability at least $1-\delta$.\\
Let $\eta > 0$, by selecting $$M = 4\eta^{-2}\norm{H}^2_F\norm{H}^2\tau,$$ we then obtain
$$\norm{AA^* - H^2} \leq \eta,$$
with probability at least $1-\delta$.\\

Now we study $\norm{\widehat{Z}_K(At,Ht) - \widehat{Z}(Ht, At)}$. Denote with $a$, $b$ the functions $a(x) = l(x^2)$, $b(x) = m(x^2)$ and with $a_K, b_K$ the associated $K$-truncated Taylor expansions. Note that $a(x) = \cos(x) - 1$, while $b(x) = (\sin(x) - x)/x$. Now by definition of $\widehat{Z}_K$ and $\widehat{Z}$, we have
\begin{align*}
\norm{\widehat{Z}_K(At,Ht) - & \widehat{Z}(Ht, At)} \nonumber \\ & \leq \norm{a_K(t \sqrt{AA^*}) - a(t \sqrt{AA^*})} + t \norm{H}\norm{b_K(t \sqrt{AA^*}) - b(t \sqrt{AA^*})}.
\end{align*}
Note that, since $$\sum_j x^{2j}/(2j)! = \cosh(x) \leq 2 e^{|x|},$$ and
\begin{align*}
|(a_K-a)(x)| &= \left|\sum_{j=K+2}(-1)^j \frac{x^{2j}}{(2j)!}\right| \\
&\leq \frac{|x|^{2K + 4}}{(2K + 4)!} \sum_{j=0} \frac{|x|^{2j}}{(2j)!} \frac{(2K + 4)!(2j)!}{(2j + 2K + 4)!} \\
&\leq \frac{2|x|^{2K + 4}e^{|x|}}{(2K + 4)!},\\
|(b_K-b)(x)| &= \left|\sum_{j=K+2}(-1)^j \frac{x^{2j}}{(2j+1)!} \right| \\
&\leq \frac{|x|^{2K + 4}}{(2K + 4)!} \sum_{j=0} \frac{|x|^{2j}}{(2j)!} \frac{(2K + 4)!(2j)!}{(2j + 2K + 5)!} \\
&\leq \frac{2|x|^{2K + 4}e^{|x|}}{(2K + 4)!}.
\end{align*}
Let $R > 0, \beta \in (0,1]$. Now note that, by Stirling approximation, $c! \geq e^{c \log \frac{c}{e}}$, so by selecting $K = \frac{e^2}{2} R + \log(\frac{1}{\beta})$, we have for any $|x| \leq R$,
\begin{align*}
\log\left|\frac{2|x|^{2K + 4}e^{|x|}}{(2K + 4)!}\right| &\leq |x| + (2K+4)\log\frac{e|x|}{2K + 4} \\
&\leq R + (2K+4)\log\frac{eR}{2K + 4} \\
& \leq R - \left(e^2 R + 4 + \log\frac{1}{\beta}\right)\log\left(e + \frac{4+\log{1}{\beta}}{eR}\right) \\
& \leq R - \left(e^2 R + 4 + \log\frac{1}{\beta}\right) \\
& \leq -(e^2 - 1)R - 4 - \log \frac{1}{\beta} \\
& \leq -\log \frac{1}{\beta}.
\end{align*}
So, by choosing $K \geq \log\frac{1}{\beta} +e^2R/2$, we have $|a_K(x) - a(x)|,|b_K(x) - b(x)| \leq \beta$.
With this we finally obtain
$$\norm{\widehat{Z}_K(At,Ht) - \widehat{Z}(Ht, At)} \leq 2(1+t\norm{H})\beta,$$
when $K \geq e^2t\norm{A}/2 + \log \frac{1}{\beta}$, and therefore we have
$$\norm{\widehat{\psi}_{K,M}(t) - \psi(t)} \leq 2t^2(1+t\norm{H})\eta + 2(1 + t\norm{H}) \beta,$$
with probability at least $1-\delta$, when
$$M \geq 8\eta^{-2}\norm{H}^2_F\norm{H}^2\log\frac{4\norm{H}^2_F}{\norm{H}^2\delta},\quad K \geq 4t\norm{A} + \log \frac{1}{\beta}.$$
In particular, by choosing $\eta = \epsilon/(4t^2(1+t\norm{H})$ and $\beta = \epsilon/(4(1+t\norm{H}))$,
we have
$$\norm{\widehat{\psi}_{K,M}(t) - \psi(t)} \leq \epsilon,$$
with probability at least $1-\delta$.

With this result in mind note that, in the event where $\norm{AA^* - H^2} \leq \epsilon$, we have that
$$|\norm{AA^*} - \norm{H}^2| \leq \norm{AA^* - H^2} \leq \epsilon,$$
and therefore, $\norm{A} \leq \sqrt{\norm{H}^2 + \epsilon}$.
\end{proof}

\section{Application to density matrix simulation}
\label{sec:applications}

Sample-based Hamiltonian simulation is a method for simulating Hamiltonians which are density matrices~\cite{lloyd2014quantum}. Such method has been used in several recent quantum machine learning algorithms, including quantum support vector machines~\cite{rebentrost2014quantum}, quantum gradient descent / Newton's method~\cite{rebentrost2016quantum}, and quantum linear regression~\cite{schuld2016prediction}.
These techniques all make use of a quantum algorithm that implements the time evolution $e^{i\rho t} \ket{\psi}$ governed by a Hamiltonian corresponding to the density matrix $\rho$ of a pure $n$-qubit state $\ket{\psi}$ of dimension $2^n$.
Specifically, given an oracle that returns superpositions of the entries of a density matrix $\rho \in \mathbb C^{2^n \times 2^n}$, the quantum algorithm for density matrix exponentiation implements a unitary matrix $U$, such that $\Vert U \ket{\psi} - e^{i \rho t} \ket{\psi}\Vert \leq \epsilon$ and requires $\Ord{t^2/\epsilon}$ copies of $\rho$. This has been shown to be optimal~\cite{kimmel2017hamiltonian}.

The total runtime of the density matrix simulation algorithm is $\Ord{t^2 T(U_{\rho})/\epsilon}$, where $T(U_{\rho})$ is the time required to prepare a quantum superposition of the entries of the matrix $\rho$. If the state preparation can be achieved in $\Ord{\poly(n)}$ time, the total runtime of the algorithm reduces to $\Ord{\poly(n) t^2/\epsilon}$, which is polylogarithmic in the dimension of $\rho$. A data structure that grants this sort of fast access is common in the quantum machine learning literature~\cite{biamonte2017quantum, ciliberto2018quantum}. As discussed in the introduction, it was recently noted by Tang~\cite{tang2018quantum} that some common implementations of this data structure, such as~\cite{kerenidis2016quantum}, are equivalent to the ability to perform $\ell_2$-norm sampling on the entries of $\rho$. This requirement is fundamentally equivalent to the row-searchable condition used in this paper.

More formally, the sample-based Hamiltonian simulation Theorem states that
\begin{theorem}[Sample-based Hamiltonian simulation~\cite{lloyd2014quantum}]
Given access to multiple copies of an $n$-qubit state $\rho$,
there is an efficient quantum circuit implementing a unitary $U$, such that $\Vert U \psi - e^{i \rho t} \psi \Vert \leq \epsilon$, where $0 \leq \epsilon \leq 1$, $t>0$, and $\psi$ is an arbitrary pure state. The algorithm requires $\Ord{ t^2/\epsilon}$ copies of $\rho$.
\end{theorem}
Note that the full running time of the algorithm depends on the memory access mode and the time to prepare $\rho$.

Theorem~\ref{thm:main_psd}, i.e. the algorithm for positive semidefinite matrices, suggests a straightforward classical analogue of sample-based Hamiltonian simulation which outputs a classical description of the evolved quantum state,
under the condition that we can efficiently compute marginals of the diagonal entries of $\rho$.
Note that this is the case for any density matrix with structured diagonal entries (\emph{e.g.}, non-decreasing order).
More precisely, we obtain a classical efficient algorithm for density matrix simulation with the following properties (the next result is a formal statement of Theorem~\ref{th:sampleinformal}):
\begin{theorem}[Classical sample-based Hamiltonian simulation]\label{thm:sample_based_Ham_sim}
Let $\epsilon, \delta \in (0,1]$, and let $\rho$ be a row-searchable $n$-qubit density matrix with at most $s$ non-zero entries in each row.
Given an efficient classical description of an arbitrary input state $\psi$ then, with probability $1-\delta$, Algorithm~\ref{alg:Nystrom} 
returns an approximation of any chosen amplitude of the state $\widehat{\psi}$
such that
$\norm{\widehat{\psi} - e^{i\rho t} \psi} \leq \epsilon$. The algorithm runs in time
$\Ord{\frac{s t^2}{\epsilon^2} \log^2 \frac{1}{\delta}}$ and requires $\Ord{\frac{t^2}{\epsilon^2} \log^2\frac{1}{\delta}}$ of memory.
\end{theorem}

\begin{proof}[Proof of Theorem~\ref{thm:sample_based_Ham_sim}]
The Theorem is an immediate consequence of the unitarity of the trace of density matrices and Theorem~\ref{thm:main_psd} from which we know that it is possible to compute an approximation $\widehat{\psi}(t)$ of $\psi(t)$, with error at most $\epsilon$ and with probability at least $1-\delta$, in  $\Ord{\frac{s t^2\tr{H}^2}{\epsilon^2} \log^2\frac{1}{\delta}}$ time and $\Ord{\frac{t^2\tr{H}^2}{\epsilon^2} \log^2\frac{1}{\delta}}$ memory. 
\end{proof}

\subsubsection*{Acknowledgements} 

We acknowledge Dorit Aharonov, Leonardo Banchi, Sougato Bose, Fernando Brand\~ao, Giuseppe Carleo, Vladimir Korepin, Laura Man\v{c}inska, Tzu-Chieh Wei for helpful conversations. A.R. is supported by EPSRC and by QinetiQ. C.C. and M.P. are supported by EPSRC. S.S. is supported by the Royal Society, EPSRC, the National Natural Science Foundation of China, and the grant ARO-MURI W911NF-17-1-0304 (US DOD, UK MOD and UK EPSRC under the Multidisciplinary University Research Initiative). L.W. is supported by the Royal Society.

%\twocolumngrid
\printbibliography

\end{document}